\documentclass[english]{article}
\usepackage{authblk}
\usepackage[T1]{fontenc}
\usepackage[latin9]{inputenc}
\usepackage{geometry}
\usepackage{amsmath}
\usepackage{amsthm}
\usepackage{amssymb}
\usepackage{color}
\usepackage{graphicx}
\usepackage[sorting=none]{biblatex}
\makeatletter
\numberwithin{equation}{section}
\numberwithin{figure}{section}
\theoremstyle{plain}
\newtheorem{thm}{\protect\theoremname}
\theoremstyle{plain}
\newtheorem{prop}[thm]{\protect\propositionname}
\theoremstyle{remark}
\newtheorem{rem}[thm]{\protect\remarkname}

\makeatother

\usepackage{babel}
\providecommand{\propositionname}{Proposition}
\providecommand{\remarkname}{Remark}
\providecommand{\theoremname}{Theorem}

\begin{document}
\global\long\def\ve{\varepsilon}%
\global\long\def\R{\mathbb{R}}%
\global\long\def\Rn{\mathbb{R}^{n}}%
\global\long\def\Rd{\mathbb{R}^{d}}%
\global\long\def\E{\mathbb{E}}%
\global\long\def\P{\mathbb{P}}%
\global\long\def\bx{\mathbf{x}}%
\global\long\def\vp{\varphi}%
\global\long\def\ra{\rightarrow}%
\global\long\def\smooth{C^{\infty}}%
\global\long\def\Tr{\mathrm{Tr}}%
\global\long\def\bra#1{\left\langle #1\right|}%
\global\long\def\ket#1{\left|#1\right\rangle }%
\global\long\def\Re{\mathrm{Re}}%
\global\long\def\Im{\mathrm{Im}}%
\global\long\def\bsig{\boldsymbol{\sigma}}%
\global\long\def\btau{\boldsymbol{\tau}}%
\global\long\def\bmu{\boldsymbol{\mu}}%
\global\long\def\bx{\boldsymbol{x}}%
\global\long\def\bups{\boldsymbol{\upsilon}}%
\global\long\def\bSig{\boldsymbol{\Sigma}}%
\global\long\def\bt{\boldsymbol{t}}%
\global\long\def\bs{\boldsymbol{s}}%
\global\long\def\by{\boldsymbol{y}}%
\global\long\def\brho{\boldsymbol{\rho}}%
\global\long\def\ba{\boldsymbol{a}}%
\global\long\def\bb{\boldsymbol{b}}%
\global\long\def\bz{\boldsymbol{z}}%
\global\long\def\bc{\boldsymbol{c}}%
\global\long\def\balpha{\boldsymbol{\alpha}}%
\global\long\def\bbeta{\boldsymbol{\beta}}%

\newcommand{\bfemph}[1]{\textbf{\textit{#1}}}
\newcommand{\ML}[1]{\textcolor{blue}{[ML:#1]}}
\newcommand{\JC}[1]{\textcolor{red}{[JC: #1]}}

\title{Direct interpolative construction of the discrete Fourier transform as a matrix product operator}
\author[1]{Jielun Chen}
\author[2]{Michael Lindsey}
\affil[1]{Department of Physics, California Institute of Technology, Pasadena, CA 91125 USA}
\affil[2]{Department of Mathematics, University of California, Berkeley, CA 94720 USA}
\renewcommand\Authand{, }
\renewcommand\Affilfont{\itshape\small}
\maketitle

\begin{abstract}
    The quantum Fourier transform (QFT), which can be viewed as a reindexing of the discrete Fourier transform (DFT), has been shown to be compressible as a low-rank matrix product operator (MPO) or quantized tensor train (QTT) operator \cite{Chen_2023_QFT}. However, the original proof of this fact does not furnish a construction of the MPO with a guaranteed error bound. Meanwhile, the existing practical construction of this MPO, based on the compression of a quantum circuit, is not as efficient as possible. We present a simple closed-form construction of the QFT MPO using the interpolative decomposition, with guaranteed near-optimal compression error for a given rank. This construction can speed up the application of the  QFT and the DFT, respectively, in quantum circuit simulations and QTT applications. We also connect our interpolative construction to the approximate quantum Fourier transform (AQFT) by demonstrating that the AQFT can be viewed as an MPO constructed using a different interpolation scheme.
\end{abstract}

\section{Introduction}
The quantum Fourier transform (QFT)~\cite{Shor_1994} is one of the most important algorithmic primitives in quantum computing. It is often viewed as providing exponential speedup compared to the fast Fourier transform (FFT), due to its gate complexity of $O(n^2)$ acting on an $n$-qubit quantum state. Such a state can be viewed as an $N$-dimensional vector with $N=2^n$, and under this identification the QFT corresponds to a discrete Fourier transform (DFT), which costs $O(N\log N) = O(2^n n)$ operations via the FFT.

However, if one makes use of low-rank structure of the input state, the QFT can often be simulated classically using efficient operations with matrix product states (MPS) and matrix product operators (MPO)~\cite{Fannes_1992, Klumper_1992, Ostlund_1995, Vidal_2003, Pirvu_2010}. Equivalently, the DFT of functions compressible as quantized tensor trains (QTT)~\cite{oseledets2011tensortrain, Khoromskij_2009_QTT} can be computed with exponential speedup. For example, one can remove the long-range quantum gates in the QFT circuit that are close to identity, resulting in the approximate QFT (AQFT)~\cite{Coppersmith_2002}. The AQFT automatically admits an MPO representation due to the bounded interaction length of the circuit~\cite{Aharonov_2007, Browne_2007, Yoran_2007}, and therefore it can be simulated classically on any input states represented as an MPS. The full QFT circuit can also be viewed as the composition of $n$ MPOs of rank 2, which can be applied sequentially to MPS inputs with intermediate compressions~\cite{Shinaoka_2023, Dolgov_2012}. This approach has been called the ``superfast Fourier transform'' or QTT-FFT \cite{Dolgov_2012}. Recently, the entire QFT operator was proven to be compressible as an MPO~\cite{Chen_2023_QFT}. In practice, the MPO representation of the QFT, referred to as the QFT MPO, can be obtained by contracting and compressing the entire QFT circuit. Interestingly, for a given error tolerance, the QFT MPO was observed to have smaller ranks compared to the AQFT MPO~\cite{Chen_2023_QFT, Woolfe_2014}, implying that the long-range gates actually play a role in reducing the entanglement. 

Classical simulation of the QFT can be used as a classical algorithm that outperforms the FFT in restricted settings. Concretely, one converts an $N$-dimensional vector into an $n$-site MPS and contracts it with the QFT circuit or the QFT MPO. In this approach, even if the state admits compression as an MPS, the conversion dominates the time complexity if the state is formed as a dense vector and compressed via the TT singular value decomposition (TT-SVD) \cite{oseledets2011tensortrain}. A randomized variant of TT-SVD yields an $O(\log N)$ speed-up \cite{Chen_2023_QFT}, and other heuristic approaches such as TT-cross \cite{Dolgov_2020_cross, Savostyanov_2014_cross} can yield $O(N)$ speedup, enabling the QTT-FFT of~\cite{Dolgov_2012}. If the input vector is the quantized representation of a function that is well approximated in a multiresolution polynomial basis, a direct analytical construction based on the interpolative decomposition can be applied with quantitative control over the error~\cite{lindsey2023multiscale}. The interpolating cores introduced in~\cite{lindsey2023multiscale} are an essential ingredient in this work. In some situations, the input is automatically available in low-rank tensor format, such as the evolving state of a differential equation solved within the QTT format~\cite{khoromskij2014tensor, Gourianov_2022, Ye_2022_vlasov, kiffner2023tensor} or a quantum simulation solved using an MPS representation 
 of the state~\cite{Daley_2004, White_2004, Vidal_2007}.

In these applications, an efficient and accurate construction of the QFT MPO is crucial. While it can be constructed efficiently by contracting the QFT circuit, this approach has a few drawbacks. First, to avoid exponential scaling in $n$, one must perform intermediate compressions before the final MPO is formed, and no rigorous \emph{a priori} estimate of the accumulated error is available, even though the approach is highly accurate in practice and heuristic arguments support its accuracy. Second, the circuit contraction takes $O(r^3 n^2)$ time, where $r$ is the bond dimension, while the MPO ultimately contains at most $O(r^2 n)$ information, so the computational complexity does not appear to be optimal, as we confirm in this work. Third, the parameters in the circuit decay like $O(1/N)$, so for large $N=2^n$ we unavoidably lose accuracy in the long-range gates due to the limitations of numerical precision. This loss of accuracy resembles the approximation imposed by the AQFT circuit, which increases the bond-dimension by the aforementioned observation in \cite{Chen_2023_QFT, Woolfe_2014}.
 
In this paper, we present a simple closed-form construction of the QFT MPO, avoiding circuit contraction and eliminating all three drawbacks mentioned above. The construction involves an interpolative decomposition based on polynomial interpolation on the Chebyshev-Lobatto grid, applied recursively on each qubit, and yields a nearly optimal MPO compression for a given bond dimension. The bond dimension of the construction corresponds to the number of interpolation points. For a fixed number of qubits, the error decreases super-exponentially with respect to the bond dimension, matching the bounds for the Schmidt coefficients derived in~\cite{Chen_2023_QFT}. Furthermore, the tensor cores are all identical (excepting the first and last). Finally, we show that the AQFT circuit can also be interpreted as an MPO constructed in the same fashion but with a different interpolation scheme, i.e., piecewise constant interpolation. This interpolation scheme converges more slowly with respect to the number of interpolation points, explaining the observation that the AQFT has a higher MPO rank than the QFT MPO, for a given error tolerance.

We outline the paper as follows. In Section~\ref{sec:prelim} we present some preliminary definitions and notation. In Section~\ref{section:rank_bound} we demonstrate how the interpolative decomposition yields a rank bound for the unfolding matrices of the QFT, similar to that of~\cite{Chen_2023_QFT}. (In fact, our rank bound is not implied by~\cite{Chen_2023_QFT} since it is defined in terms of the infinity norm error, rather than the Frobenius norm error.) We also describe the connection of the rank bound for the unfolding matrices with the \emph{complementary low-rank} structure~\cite{Li_2015_butterfly} that has been observed for the DFT and other operators.
This section can also be viewed as a warm-up for Section~\ref{sec:construction}, in which we present our construction of the QFT as an MPO using the interpolative decomposition. In Section~\ref{sec:error}, we bound the error of this construction. In Section~\ref{sec:aqft}, we illustrate the connection to the AQFT.

\subsection{Acknowledgments}
The authors are grateful to Sandeep Sharma, Garnet Chan, and Yuehaw Khoo for helpful discussions.

\section{Preliminaries}
\label{sec:prelim}

Consider discrete indices $s,t\in\{0,1,\ldots,N-1\}$, where $N=2^{n}$,
and define the matrix of the \bfemph{discrete Fourier transform (DFT)}
\[
\mathbf{F}_{s,t}=e^{-\frac{2\pi ist}{N}}.
\]
We comment that in the quantum computing literature, the convention for the QFT is $\exp(2\pi i s t/N)$ rather than $\exp(-2\pi i s t/N)$, but we use the latter to match the standard definition of the DFT.

We can place $s$ and $t$, respectively, in bijection with bit strings
$\sigma_{1:n}=(\sigma_{1},\ldots,\sigma_{n})$ and $\tau_{1:n}=(\tau_{1},\ldots,\tau_{n})$
in $\{0,1\}^{n}$ via 
\begin{equation}
\label{eq:bindec}
s=\sum_{k=1}^{n}2^{n-k}\sigma_{k},\quad t=\sum_{k=1}^{n}2^{k-1}\tau_{k}.
\end{equation}
These bits can be viewed as binary digits for $s$ and $t$ with opposite orderings.

Then we can view $\mathbf{F}$ as a tensor $F = F(\sigma_{1:n}, \tau_{1:n})$ of size $2^{2n}$ via this identification: 
\begin{equation}
\label{eq:DFT}
\mathbf{F}_{s,t}=F(\sigma_{1:n},\tau_{1:n})=e^{-\pi i\sum_{k,l=1}^{n}2^{-k}2^{l}\sigma_{k}\tau_{l}}.
\end{equation}

In the QFT, $\tau_k$ and $\sigma_k$ correspond to the input and output indices of the $k$-th qubit respectively. Notice that $\tau_1$ is the least significant bit while $\sigma_1$ is the most significant bit, and it is a key point in~\cite{Chen_2023_QFT} that only this ordering induces low-rank structure of the QFT as a \bfemph{matrix product operator} or \bfemph{quantized tensor train} operator, which is defined as the approximate factorization
\begin{equation}
    F(\sigma_{1:n}, \tau_{1:n}) \approx \sum_{\alpha_1\in [r_1],\dots,\alpha_{n-1}\in [r_{n-1}]} A_1^{\alpha_1}(\sigma_{1},\tau_{1}) A_2^{\alpha_1, \alpha_2}(\sigma_{2},\tau_{2}) \dots A_n^{\alpha_{n-1}}(\sigma_{n},\tau_{n}).
\end{equation}
where we used the notation $[r_k] = \{0,\dots,r_k-1\}$. The integers $r_k$ are often called the \bfemph{bond dimensions} or \bfemph{TT ranks}, and for simplicity we assume that $r_k = r$ for some fixed $r$. The tensors $A_{k}^{\alpha_{k-1}, \alpha_{k}}(\sigma_k, \tau_k)$ are referred to as \bfemph{tensor cores}.

We also define the $m$-th \bfemph{unfolding matrix} of $F$ as the $2^{2m} \times 2^{2(n-m)}$ matrix
\begin{equation}
    T_m( {\sigma_{1:m},\tau_{1:m}} ;  {\sigma_{m+1:n}, \tau_{m+1:n}}) = F(\sigma_{1:n}, \tau_{1:n}),
\end{equation}
where we view the multi-indices before and after the semicolon, respectively, as row and column indices of a matrix. It is shown in \cite{oseledets2011tensortrain} that a QTT with good approximation exists if the $m$-th unfolding matrix is numerically low-rank for all $m$.

To measure the error of our decompositions it is useful to define a tensor \bfemph{infinity norm}
\begin{equation}
    \|T\|_\infty = \Vert \mathrm{vec}(T) \Vert_\infty
\end{equation}
as the largest magnitude of any element of the tensor. We will use this notation for tensors of different shapes, but ultimately the accuracy of our MPO approximation itself is most naturally expressed in this norm.


We also define notations for our interpolation scheme. We will denote the  \bfemph{Chebyshev-Lobatto grid points} (shifted and scaled to the interval $[0, 1]$) by 
\begin{equation}
    c^{\alpha} = \frac{1 - \cos(\pi \alpha/K) }{2}
\end{equation}
for $\alpha = 0, ..., K$. Throughout, we use $K$ to denote grid size (or more properly, the grid size minus one). The corresponding Chebyshev-Lobatto interpolation of a function $f(x)$, i.e., the Lagrange polynomial interpolation on the Chebyshev-Lobatto grid, is 
\begin{equation}
    f(x) \approx \sum_{\alpha = 0}^K f(c^\alpha) P^{\alpha}(x),
\end{equation}
where the $P^\alpha$ are degree-$K$ polynomials such that $P^\alpha(c^\beta) = \delta_{\alpha, \beta}$, i.e., the Lagrange interpolating polynomials for this grid.
In the case of Chebyshev-Lobatto interpolation, the $P^\alpha$ are sometimes called the \bfemph{Chebyshev cardinal functions}, which can be evaluated directly with simple trigonometric formulas \cite{Boyd_1992}, bypassing the application of the more general formula for Lagrange basis functions.

\section{Rank bound}
\label{section:rank_bound}
Although the rank of $F(\sigma_{1:n},\tau_{1:n})$ as an MPO was already
bounded in \cite{Chen_2023_QFT}, we provide an alternative rank bound from an
interpolative perspective. In fact, this point of view yields \emph{entrywise}
control over the error of $F$, which is not captured by the bounds on the Schmidt spectrum in \cite{Chen_2023_QFT}.

Our motivation arises from computing 
\begin{align*}
F(\sigma_{1:n},\tau_{1:n}) & =e^{-\pi i\sum_{k,l=1}^{n}2^{-k}2^{l}\sigma_{k}\tau_{l}}\\
 & =e^{-\pi i\sum_{k,l=1}^{m}2^{-k}2^{l}\sigma_{k}\tau_{l}}\times e^{-\pi i\sum_{k=1}^{m}\sum_{l=m+1}^{n}2^{-k}2^{l}\sigma_{k}\tau_{l}}\\
 & \quad\quad\quad\times\ e^{-\pi i\sum_{k=m+1}^{n}\sum_{l=1}^{m}2^{-k}2^{l}\sigma_{k}\tau_{l}}\times e^{-\pi i\sum_{k,l=m+1}^{n}2^{-k}2^{l}\sigma_{k}\tau_{l}}.
\end{align*}
 Now if $k \leq m$ and $l > m$, then $2^{-k}2^{l}\sigma_{k}\tau_{l}\in2\mathbb{Z}$,
and therefore 
\[
e^{-\pi i\sum_{k=1}^{m}\sum_{l=m+1}^{n}2^{-k}2^{l}\sigma_{k}\tau_{l}}=1.
\]
 It follows that we can write 
\[
F(\sigma_{1:n},\tau_{1:n})=F(\sigma_{1:m},\tau_{1:m})F(\sigma_{m+1:n},\tau_{m+1:n})
e^{-2\pi i\left(\sum_{k=m+1}^{n}2^{m-k}\sigma_{k}\right)\left(\sum_{l=1}^{m}2^{l-m-1}\tau_{l}\right)},
\]
where we have overloaded the notation for $F$ slightly to denote the quantized tensor representations of DFT matrices of arbitrary sizes.

It is then useful to define 
\begin{equation}
x_{>m}=\sum_{k=m+1}^{n}2^{m-k}\sigma_{k}\in[0,1],\quad y_{\leq m}=\sum_{l=1}^{m}2^{l-m-1}\tau_{l}\in[0,1],\label{eq:xgtr_yless}
\end{equation}
 so 
\begin{equation}
F(\sigma_{1:n},\tau_{1:n})=F(\sigma_{1:m},\tau_{1:m})F(\sigma_{m+1:n},\tau_{m+1:n})e^{-2 \pi ix_{>m}y_{\leq m}}. \label{eq:QFT_decomp}
\end{equation}

Then we are motivated to consider, for arbitrary $y\in[0,1]$, the
function 
\[
f_{y}(x)=e^{-2\pi ixy},\quad x\in[0,1].
\]
 We can approximate $f_{x}$ with Chebyshev-Lobatto interpolation
to obtain 
\[
f_{y}(x)\approx\sum_{\alpha=0}^{K}e^{-2\pi iyc^{\alpha}}P^{\alpha}(x).
\]
 If this equality holds with pointwise error bounded by $\ve$, then the rank-$(K+1)$ approximation of the $m$-th unfolding
matrix of $F$ has entrywise error bounded by $\ve$.

The above arguments are reviewed diagrammatically in Fig.~\ref{fig:QFT_decomp}. Next we state and prove an error bound.

\begin{prop}
There exists a rank-$(K+1)$ approximation of the $m$-th unfolding
matrix of $F$ whose entrywise error is bounded uniformly by 
\[
\frac{4\left(\frac{\pi}{2}\right)^{K+1}e^{K}K^{-K}}{K-\frac{\pi}{2}}.
\]
\label{prop:unfolding_mat_error}
\end{prop}

\begin{rem}
This bound on the infinity norm error implies a bound on the Frobenius norm error very similar to the Frobenius error implied by the Schmidt coefficient bound in~\cite{Chen_2023_QFT}, but with a worse exponential prefactor $c^K$. However, the $K^{-K}$ decay present in both bounds quickly takes over. Thus Proposition~\ref{prop:unfolding_mat_error} has similar implications as the main theorem in the context of~\cite{Chen_2023_QFT}. Meanwhile, if entrywise control is needed, then Proposition~\ref{prop:unfolding_mat_error} has additional value.
\end{rem}

\begin{figure}
    \centering
    \includegraphics[width=0.9\textwidth]{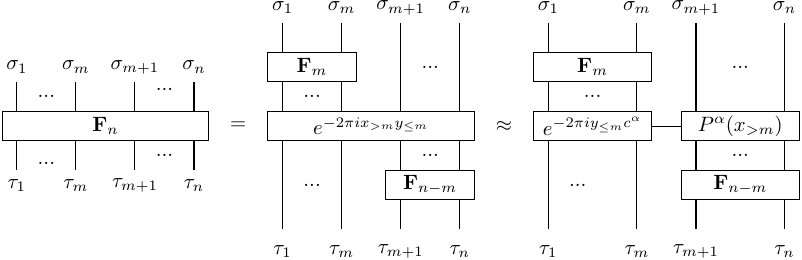}
    \caption{Diagrammatic illustration of the arguments of Section~\ref{section:rank_bound}. The first equality in the figure depicts~\eqref{eq:QFT_decomp}, and the second approximate equality depicts the low-rank decomposition of the $m$-th unfolding matrix, which derives from the Chebyshev-Lobatto interpolation with entrywise error controlled by Proposition~\ref{prop:unfolding_mat_error}.}
    \label{fig:QFT_decomp}
\end{figure}

\begin{proof}[Proof of Proposition~\ref{prop:unfolding_mat_error}]
To apply standard results on Chebyshev interpolation error, we rescale
$f_{y}$ to the reference interval $[-1,1]$, defining 
\[
f_{y}(z)=f_{y}\left(\frac{z+1}{2}\right)=e^{-\pi iy}e^{-\pi iyz}.
\]
 Note that $g_{y}$ extends analytically to all of $z$. For $\rho > 1$, define the Bernstein ellipse
$\mathcal{E}_{\rho}$ by
\[
    \mathcal{E}_{\rho} := \left\{z \in \mathbb{C} : \left[\frac{\Re(z)}{a_\rho}\right]^2 + \left[\frac{\Im(z)}{b_\rho}\right]^2 \leq 1 \right\}
\]
where
\[
    a_\rho := \frac{\rho + \rho^{-1}}{2}, \quad b_\rho := \frac{\rho - \rho^{-1}}{2}.
\]
Moreover,
$\vert g_{y}(z)\vert=e^{\pi y\,\Im(z)}$, so on the Bernstein ellipse $\mathcal{E}_{\rho}$ we have that 
\[
\vert g_{y}(z)\vert\leq e^{\pi y(\rho-\rho^{-1})/2}\leq e^{\pi y\rho/2}.
\]
 Then by Theorem 8.2 of~\cite{TrefethenBook2019}, it follows that the pointwise
error of Chebyshev interpolation of $g_{y}$ is bounded (independently
of $y\in[0,1]$) by 
\[
\frac{4e^{\pi\rho/2}\rho^{-K}}{\rho-1}
\]
 for any $\rho>1$. We want to choose $\rho$ to optimize this bound.
It is roughly equivalent to find the optimizer of the numerator, which
is simple to compute exactly as $\rho=\frac{2}{\pi}K$. Then plugging
into our expression yields the bound 
\[
\frac{4\left(\frac{\pi}{2}\right)^{K+1}e^{K}K^{-K}}{K-\frac{\pi}{2}},
\]
which completes the proof.
\end{proof}

\subsection{Connection to complementary low-rankness}
The low-rankness of the unfolding matrices in fact implies the well-known \emph{complementary low-rank} property \cite{Li_2015_butterfly} of the DFT matrix. This property says that, if we partition the matrix into $2^l \times 2^{n-l}$ blocks of size $2^{n-l} \times 2^l$ where $l$ is the level of the partition, then for all levels $l=0,\ldots, n$ each block is \emph{numerically low-rank}. For example, when $n = 3$, the partition can be visualized as follows.
\begin{center}
    \includegraphics[width=0.75\textwidth]{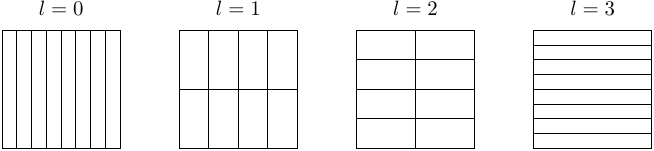}
\end{center}

We argue that the low-rankness of the unfolding matrix is a stronger property than complementary low-rankness. Suppose that we fix the level $l$ and let $F_{i, j}$ denote the $(i,j)$-th block of the DFT matrix at this level. The complementary low-rank property implies that a decomposition exists for all $(i, j)$, i.e., $F_{i,j} \approx \tilde{R}_{i,j} \tilde{L}_{i,j} $ for some tensors $\tilde{R}$ and $\tilde{L}$. On the other hand, the low-rankness of the unfolding matrix implies the more special structure $F_{i,j} \approx R_j L_i$.

We illustrate this point diagrammatically as follows.
\begin{center}
    \includegraphics[width=0.55\textwidth]{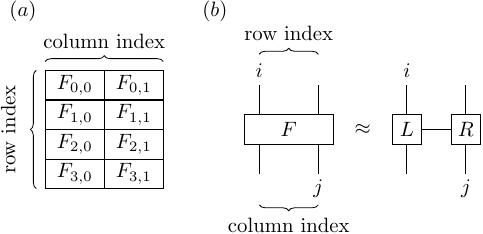}
\end{center}
Here $(a)$ depicts the partition and the block indices for $n = 3$ and $l = 2$, and $(b)$ describes the low-rank approximation of the unfolding matrix, which is equivalent to $F_{i,j} \approx R_{j} L_{i}$. It is indeed known~\cite{Edelman1997TheFF} that the submatrices for the DFT can be written as $F_{i,j} = D_j F_{0, 0} D_i$ where $D_i$ (resp., $D_j$) is the unitary diagonal matrix generated by the $i$-th (resp., $j$-th) column of the $2^{l} \times 2^{l}$ (resp., $2^{n-l} \times 2^{n-l}$) DFT. This point is essentially a restatement of~\eqref{eq:QFT_decomp}.

\section{Direct interpolative construction}
\label{sec:construction}

The preceding argument does not directly furnish an approximate presentation
of $F$ as an MPO. Now we turn to the task of such a construction, which is illustrated diagrammatically in Fig.~\ref{fig:QFT_recur}. 

\begin{figure}
    \centering
    \includegraphics[width=0.75\textwidth]{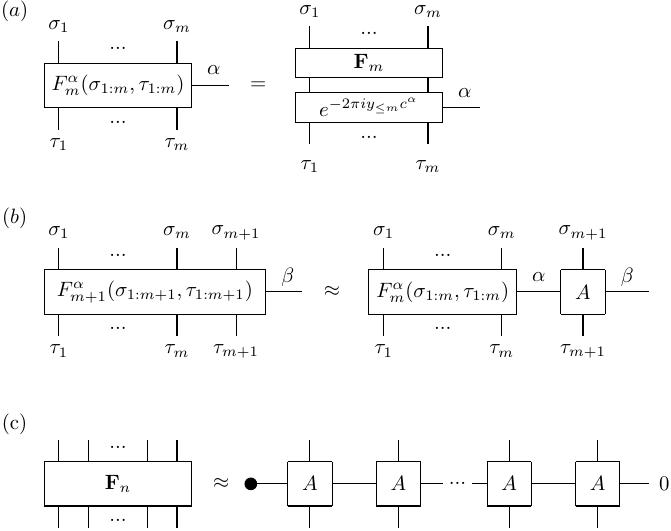}
    \caption{Illustration of the QFT MPO construction. (a) The $m$-th tensor $F_m$ to be interpolated, cf.~\eqref{eq:Falpha}. (b) We recursively obtain $F_{m+1}$ from $F_m$ by attaching the tensor core $A$ defined in~\eqref{eq:core}. (c) The full MPO of the DFT with translational-invariant cores. The black tensor represents an all-one vector. The right and left `boundary' conditions follow from Eqs.~\eqref{eq:AR} and \eqref{eq:AL}, respectively.}
    \label{fig:QFT_recur}
\end{figure}

Recall our definition (\ref{eq:xgtr_yless}) for $y_{\leq m}$ in
terms of the bit substring $\tau_{1:m}$, which we reproduce here
as
\[
y_{\leq m}=\sum_{l=1}^{m}2^{l-m-1}\tau_{l},
\]
 for every $m=1,\ldots,n$. Note that $y_{\leq k}\in[0,1]$ for all
$k$, and moreover we can construct the sequence $y_{\leq m}$ recursively
via the relation
\[
y_{\leq m+1}=\frac{1}{2}y_{\leq m}+\frac{1}{2}\tau_{m+1},
\]
 which holds for $m=1,\ldots,n-1$. Note that $y_{\leq1}=\frac{1}{2}\tau_{1}$. 

We will construct inductively a sequence of tensors $F_{m}\in\mathbb{C}^{2^{m}\times2^{m}\times(K+1)}$,
$m=1,\ldots,n$, with entries specified by 
\begin{equation}
\label{eq:Falpha}
F_{m}^{\alpha}(\sigma_{1:m},\tau_{1:m})=e^{-\pi i\sum_{k,l=1}^{m}2^{-k}2^{l}\sigma_{k}\tau_{l}}e^{-2\pi iy_{\leq m}c^{\alpha}},\quad\sigma_{1:m},\tau_{1:m}\in\{0,1\}^{m},\ \alpha\in[K],
\end{equation}
where we point out that $y_{\leq m}$ is always understood to be
implicitly defined in terms of $\tau_{1:m}$.

For the base case, the tensor $F_{1}\in\mathbb{C}^{2\times2\times(K+1)}$
can be written explicitly
\[
F_{1}^{\beta}(\sigma,\tau)=e^{-\pi i(\sigma+c^{\beta})\tau},\quad\sigma,\tau\in\{0,1\},\ \beta\in[K].
\]
 Alternatively, we may view $F_{1}^{\beta}$ for each $\beta$ as
the $2\times2$ matrix 
\[
F_{1}^{\beta}=\left(\begin{array}{cc}
1 & e^{-\pi\iota c^{\beta}}\\
1 & -e^{-\pi\iota c^{\beta}}
\end{array}\right).
\]
 Moreover, we can obtain our target $F(\sigma_{1:n},\tau_{1:n})$
as 
\[
F(\sigma_{1:n},\tau_{1:n})=F_{n}^{0}(\sigma_{1:n},\tau_{1:n}),
\]
 under the convention that $\alpha=0$ is the index of the leftmost
node $c^{\alpha}=0$.

Then we claim that $F_{m+1}$ can be obtained approximately from $F_{m}$
by attaching a single tensor core on the right. To see
this, expand: 
\begin{align*}
F_{m+1}^{\beta}(\sigma_{1:m+1},\tau_{1:m+1}) & =e^{-\pi i\sum_{k,l=1}^{m+1}2^{-k}2^{l}\sigma_{k}\tau_{l}}e^{-2\pi iy_{\leq m+1}c^{\beta}}\\
 & =e^{-\pi i\sum_{k,l=1}^{m}2^{-k}2^{l}\sigma_{k}\tau_{l}}\times e^{-\pi i\sum_{k=1}^{m}2^{-k}2^{m+1}\sigma_{k}\tau_{m+1}}\\
 & \quad\quad\quad\times\ e^{-\pi i\sigma_{m+1}\sum_{l=1}^{m}2^{l-m-1}\tau_{l}}\times e^{-\pi i\sigma_{m+1}\tau_{m+1}}\times e^{-2\pi iy_{\leq m+1}c^{\beta}}.
\end{align*}
 Now for any $k=1,\ldots,m$, observe that $2^{-k}2^{m+1}\sigma_{k}\in2\mathbb{Z}$,
so the second factor in the last expression is just $1$. Moreover,
we can substitute $\sum_{l=1}^{m}2^{l-m-1}\tau_{l}=y_{\leq m}$ in
the third factor and $y_{\leq m+1}=\frac{1}{2}y_{\leq m}+\frac{1}{2}\tau_{m+1}$
in the last factor to obtain 
\begin{align*}
F_{m+1}^{\beta}(\sigma_{1:m+1},\tau_{1:m+1}) & =e^{-\pi i\sum_{k,l=1}^{m}2^{-k}2^{l}\sigma_{k}\tau_{l}}e^{-2\pi iy_{\leq m}\left(\frac{\sigma_{m+1}+c^{\beta}}{2}\right)}e^{-\pi i\tau_{m+1}(\sigma_{m+1}+c^{\beta})}\\
 & =e^{-\pi i\sum_{k,l=1}^{m}2^{-k}2^{l}\sigma_{k}\tau_{l}}f_{y_{\leq m}}\left(\frac{\sigma_{m+1}+c^{\beta}}{2}\right)e^{-\pi i\tau_{m+1}(\sigma_{m+1}+c^{\beta})},
\end{align*}
 where we retain our definition of $f_{y}:[0,1]\ra\R$ from the last
section.

Then we can insert an interpolative decomposition 
\begin{equation}
\label{eq:interperror}
f_{y_{\leq m}}\left(\frac{\sigma_{m+1}+c^{\beta}}{2}\right)\approx\sum_{\alpha=0}^{K}f_{y_{\leq m}}(c^{\alpha})P^{\alpha}\left(\frac{\sigma_{m+1}+c^{\beta}}{2}\right)=\sum_{\alpha=0}^{K}e^{-2\pi iy_{\leq m}c^{\alpha}}P^{\alpha}\left(\frac{\sigma_{m+1}+c^{\beta}}{2}\right)
\end{equation}
into our expression for $F_{m+1}^{\beta}(\sigma_{1:m+1},\tau_{1:m+1})$,
yielding 
\begin{align*}
F_{m+1}^{\beta}(\sigma_{1:m+1},\tau_{1:m+1}) & \approx\sum_{\alpha=0}^{K}e^{-\pi i\sum_{k,l=1}^{m}2^{-k}2^{l}\sigma_{k}\tau_{l}}e^{-2\pi iy_{\leq m}c^{\alpha}}P^{\alpha}\left(\frac{\sigma_{m+1}+c^{\beta}}{2}\right)e^{-\pi i\tau_{m+1}(\sigma_{m+1}+c^{\beta})}\\
 & =\sum_{\alpha=0}^{K}F_{m}^{\alpha}(\sigma_{1:m},\tau_{1:m})A^{\alpha\beta}(\sigma_{m+1},\tau_{m+1}),
\end{align*}
 where we have defined the tensor core $A\in\mathbb{C}^{2\times2\times(K+1)\times(K+1)}$
via 
\begin{equation}
A^{\alpha\beta}(\sigma,\tau)=P^{\alpha}\left(\frac{\sigma+c^{\beta}}{2}\right)e^{-\pi i(\sigma+c^{\beta})\tau},\quad\sigma,\tau\in\{0,1\},\ \alpha,\beta\in[K].\label{eq:core}
\end{equation}
Alternatively we can write 
\[
A^{\alpha\beta}(\sigma,\tau)=P^{\alpha}\left(\frac{\sigma+c^{\beta}}{2}\right)F_{1}^{\beta}(\sigma,\tau).
\]

Then we have argued that we can approximate $F$ as an MPO with leftmost
core $A_{\mathrm{L}}=F_{1}$, internal cores $A$, and rightmost core
$A_{\mathrm{R}}\in\R^{2\times2\times(K+1)}$ given by 
\begin{equation}
\label{eq:AR}
A_{\mathrm{R}}^{\alpha}(\sigma,\tau)=A^{\alpha0}(\sigma,\tau).
\end{equation}
Since $\sum_{\alpha = 0}^{K} P^{\alpha}(x) = 1$ for all $x \in [0, 1]$, we can neatly express the leftmost core in terms of $A$ as 
\begin{equation}
\label{eq:AL}
A_L = \sum_{\alpha = 0}^{K} A^{\alpha \beta}. 
\end{equation}

\begin{rem}
It is straightforward to extend the scheme to qudits. Redefining $\sigma_k$ and $\tau_k$ to be elements of $\{0, 1, ..., d-1\}$, we consider the $d$-ary expansions  
\[
s=\sum_{k=1}^{n}d^{n-k}\sigma_{k},\quad t=\sum_{k=1}^{n}d^{k-1}\tau_{k}.
\]
in place of~\eqref{eq:bindec}.

In this setting the DFT matrix can be written as
\[
\mathbf{F}_{s,t}=F(\sigma_{1:n},\tau_{1:n})=e^{-\frac{2}{d}\pi i\sum_{k,l=1}^{n}d^{-k}d^{l}\sigma_{k}\tau_{l}}.
\]
Following the same construction, we obtain the core
\[
A^{\alpha\beta}(\sigma,\tau)=P^{\alpha}\left(\frac{\sigma+c^{\beta}}{d}\right)e^{-\frac{2}{d}\pi i(\sigma+c^{\beta})\tau},\quad\sigma,\tau\in\{0,1,...,d-1  \},\ \alpha,\beta\in[K].
\]
\end{rem}

\section{Error bound}
\label{sec:error}

The following result bounds the entrywise error of the MPO constructed
as above.
\begin{thm}
\label{thm:mpo}
For fixed integer $K \geq 1$, let $F_{\mathrm{MPO}}$ denote the $n$-site, rank-$(K+1)$ MPO comprised of 
internal cores
\[
A^{\alpha\beta}(\sigma,\tau)=P^{\alpha}\left(\frac{\sigma+c^{\beta}}{2}\right)e^{-\pi i(\sigma+c^{\beta})\tau},\quad\sigma,\tau\in\{0,1\},\ \alpha,\beta\in[K],
\]
leftmost core $A_{\mathrm{L}}^\beta = \sum_{\alpha=0}^K A^{\alpha\beta} = e^{-\pi i(\sigma+c^{\beta})\tau}$ and rightmost core $A_{\mathrm{R}}^\alpha = A^{\alpha0} = P^{\alpha}(\sigma/2) e^{-\pi i\sigma\tau}$, as described in the
preceding discussion and illustrated in Figure~\ref{fig:QFT_recur}(c).

Then the entrywise error of this MPO with respect
to the true DFT operator $F$ is bounded as 
\[
\Vert F_{\mathrm{MPO}} - F \Vert_\infty \leq \frac{\Lambda_{K}^{n-1}-1}{\Lambda_{K}-1} E_{K} \leq (n-1)\Lambda_{K}^{n-2}E_{K},
\]
 where $\Lambda_{K}$ is the Lebesgue constant of $(K+1)$-point Chebyshev-Lobatto
interpolation and $E_{K}$ is the worst case $(K+1)$-point Chebyshev-Lobatto
interpolation error over the function class $f_{y}:[0,1]\ra\R$ defined
by $f_{y}(x)=e^{-2\pi ixy}$, $y\in[0,1]$.

The quantities $\Lambda_K$ and $E_K$ are in turn bounded as 
\[
\Lambda_{K}\leq1+\frac{2}{\pi}\log(K+1), \quad E_{K}\leq\frac{4\left(\frac{\pi}{2}\right)^{K+1}e^{K}K^{-K}}{K-\frac{\pi}{2}}.
\]
\end{thm}

\begin{proof}
Let $\tilde{F}_{m+1}=F_{1}A^{m}\in\R^{2^{m+1}\times2^{m+1}\times(K+1)}$
denote the intermediate MPO (with a hanging bond at the right) furnished
by procedure outlined above, meant to approximate $F_{m+1}$. We will use the notation $F_m A$ and $\tilde{F}_m A$ to denote the tensors obtained by attaching the tensor core $A$ to $F_m$ and $\tilde{F}_m$, respectively, as illustrated in Figure~\ref{fig:QFT_recur}.

We will inductively bound the error $\ve_{m}:=\Vert\tilde{F}_{m}-F_{m}\Vert_{\infty}$.
Note that in the base case we have $\ve_{1}=0$. Moreover, since $F = {F}_n^{0}$ and $F_{\mathrm{MPO}} = \tilde{F}_n^{0}$, we have that $\Vert F_{\mathrm{MPO}} - F \Vert_\infty \leq \ve_n$, and it will suffice to get a bound on $\ve_n$.

First compute  
\begin{align*}
\ve_{m+1} & =\Vert\tilde{F}_{m+1}-F_{m+1}\Vert_{\infty}\\
 & =\Vert\tilde{F}_{m}A-F_{m+1}\Vert_{\infty}\\
 & \leq\Vert\tilde{F}_{m}A-F_{m}A\Vert_{\infty}+\Vert F_{m}A-F_{m+1}\Vert_{\infty}\\
 & =\Vert(\tilde{F}_{m}-F_{m})A\Vert_{\infty}+\Vert F_{m}A-F_{m+1}\Vert_{\infty}.
\end{align*}
Observe that the second term in the last expression is bounded is the entrywise error of the interpolation~\eqref{eq:interperror}, which is bounded by $E_K$. In the proof of Proposition~\ref{prop:unfolding_mat_error} above,
we have already shown the bound
\[
E_{K} \leq \frac{4\left(\frac{\pi}{2}\right)^{K+1}e^{K}K^{-K}}{K-\frac{\pi}{2}}
\]
appearing in the statement of Theorem~\ref{thm:mpo}.

Now define $E_{m}:=\tilde{F}_{m}-F_{m}$, so we have that $E_{m}$
is entrywise bounded by $\ve_{m}$, and we wish to uniformly bound
the entries $[E_{m}A]^{\beta}(\sigma_{1:m},\tau_{1:m})$ of $E_{m}A$.
To this end, compute: 
\begin{align*}
\left|[E_{m}A]^{\beta}(\sigma_{1:m},\tau_{1:m})\right| & =\left|\sum_{\alpha=0}^{K}E_{m}^{\alpha}(\sigma_{1:m},\tau_{1:m})A^{\alpha\beta}(\sigma_{m+1},\tau_{m+1})\right|\\
 & \leq\ve_{m}\sum_{\alpha=0}^{K}\left|A^{\alpha\beta}(\sigma_{m+1},\tau_{m+1})\right|\\
 & =\ve_{m}\sum_{\alpha=0}^{K}\left|P^{\alpha}\left(\frac{\sigma+c^{\beta}}{2}\right)\right|\\
 & \leq\ve_{m}\Lambda_{K},
\end{align*}
 where $\Lambda_{K}$ is the \emph{Lebesgue constant }of our interpolation
scheme, cf.~\cite{TrefethenBook2019}.
It is known~\cite{TrefethenBook2019} that 
\[
\Lambda_{K}\leq1+\frac{2}{\pi}\log(K+1).
\]

Then putting our bounds together we have derived 
\[
\ve_{m+1}\leq\Lambda_{K}\ve_{m}+E_{K}.
\]
 It follows that 
\[
\ve_{n}\leq\left[\sum_{p=0}^{n-2}\Lambda_{K}^{p}\right]E_{K}=\frac{\Lambda_{K}^{n-1}-1}{\Lambda_{K}-1}E_{K}.
\]
Observe that a crude but simpler bound is given by 
\[\ve_{n}\leq(n-1)\Lambda_{K}^{n-2}E_{K}.\]
\end{proof}

\section{Connection to the approximate QFT}
\label{sec:aqft}

The AQFT \cite{Coppersmith_2002} on $n$ qubits with approximation level $b \in \{ 0,\ldots , n-1 \}$ is defined by 
\begin{equation}
\label{eq:AQFT}
    {F}_{n,b}(\sigma_{1:n}, \tau_{1:n})  = e^{-\pi i \sum_{k=1}^n \sum_{l=\max(1, k-b)}^{n} 2^{-k} 2^l \sigma_k \tau_{l}},
\end{equation}
up to the choice of convention for normalization and complex conjugation. Relative to the formula~\eqref{eq:DFT} for the QFT/DFT, in the AQFT we discard all terms in the exponent with $k-l > b$. If $b = n-1$ we recover the exact QFT, and if $b = 0$ we get the Hadamard transform. The quantum circuits for the QFT and AQFT are shown in Fig.~\ref{fig:AQFT}. We prove the following theorem regarding the connection to interpolative decomposition:

\begin{figure}
    \centering   
    \includegraphics[width=1\textwidth]{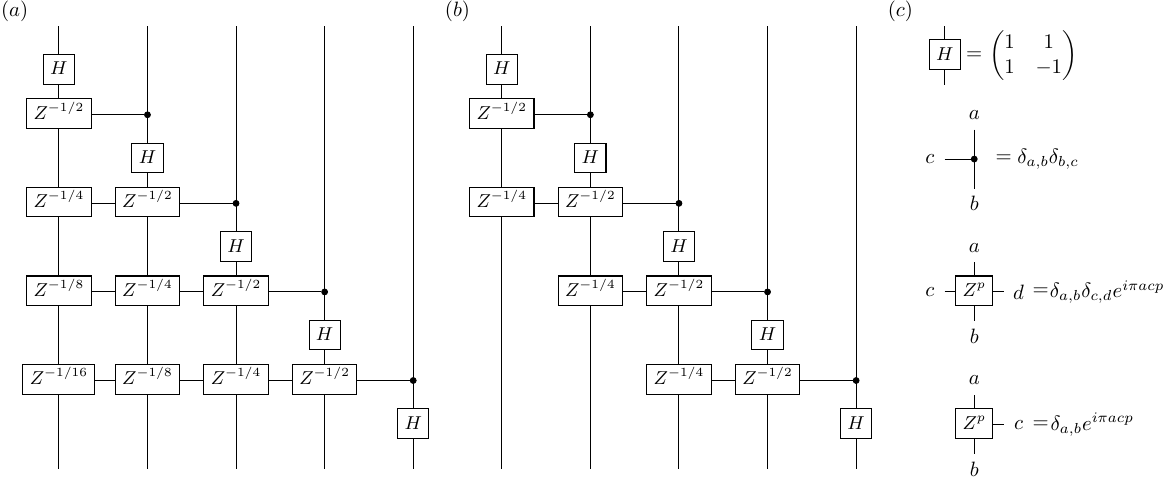}
    \caption{(a) A quantum circuit representing the 5-qubit QFT. (b) A quantum circuit representing the 5-qubit AQFT with approximation level $b=2$. (c) Tensors to interpret the circuits as tensor networks.}
    \label{fig:AQFT}
\end{figure}

\begin{thm}
The $n$-qubit AQFT with approximation level $b$ can be expressed exactly as a rank-$(2^b)$ MPO in which all of the internal cores are given by 
\begin{equation}
    \label{eq:A_AQFT}
    A^{\alpha\beta}(\sigma,\tau)=\chi^{\alpha}\left(\frac{\sigma+u^{\beta}}{2}\right)e^{-\pi i(\sigma+u^{\beta})\tau},\quad \sigma, \tau \in \{0,1\}, \ \alpha , \beta \in [2^b - 1],
\end{equation}
where $u^\beta = \beta/2^b$ for $\beta = 0,...,2^b-1$, and $\chi^\alpha(x)$ are the indicator functions
\begin{equation*}
\chi^\alpha(x) = 
\begin{cases}
1, &  x \in [u^\alpha , u^{\alpha + 1} ), \\
0, &  \mathrm{otherwise}.
\end{cases}
\end{equation*}
The leftmost core is given by $\sum_{\alpha=0}^{2^b - 1} A^{\alpha\beta}(\sigma,\tau) = e^{-\pi i(\sigma+u^{\beta})\tau}$, and the rightmost core is given by $A^{\alpha 0}(\sigma,\tau) = \chi^{\alpha}(\sigma/2) e^{-\pi i\sigma\tau}$. 
\end{thm}

\begin{proof}
    We adopt the notation $[0.\tau_p...\tau_1] = \sum_{j=1}^p \tau_{p+1-j} 2^{-j}$ for arbitrary bit strings $\tau_1 , \ldots, \tau_p$ of arbitrary length. Intuitively the notation indicates a binary decimal expansion. With this notation, for any $m$ we can express the $m$-qubit AQFT with approximation level $b$  as
    \[
        {F}_{m,b} (\sigma_{1:m}, \tau_{1:m}) = e^{-2\pi i \sum_{k=1}^m \sigma_k [0.\tau_{k}...\tau_{k-b}]}.
    \]
    Indeed, for each $k$, note that in~\eqref{eq:AQFT} all the terms with $l > k$ contribute only integer multiples of $2\pi i$ in the exponent, which all vanish.
    
    Now for any $\alpha, \beta \in [2^b -1 ]$, we can uniquely write $\alpha = \sum_{j=1}^{b} \alpha_j 2^{b-j}$ and $\beta = \sum_{j=1}^{b} \beta_j 2^{b-j}$ for bit strings $\alpha_1, \ldots, \alpha_b$ and $\beta_1, \ldots, \beta_b$. We will identify $\alpha$ and $\beta$ with these respective bit strings in the notation.
    
    Then define
    \[
        B_{m,b}^\beta (\tau_{1:m}) := e^{- 2\pi i \sum_{j=1}^b \beta_j 2^{-j} [0.\tau_{m}...\tau_{m-b+j}]},
    \]
    and
    \begin{align}
        {Q}_{m,b}^\beta (\sigma_{1:m}, \tau_{1:m}) \ := \  \ &  {F}_{m,b} (\sigma_{1:m}, \tau_{1:m}) B_{m,b}^\beta (\tau_{1:m}) \label{eq:Qdef} \\
        \  = \  \  &  e^{-2\pi i \sum_{j=1}^m \sigma_j [0.\tau_{j}...\tau_{j-b}]} e^{- 2\pi i \sum_{j=1}^b \beta_j 2^{-j} [0.\tau_{m}...\tau_{m-b+j}]}. \nonumber
    \end{align}
    Note that $B_{m,b}^0 \equiv 1$, and thus by setting $\beta = 0$ we recover the AQFT, i.e. ${Q}_{m,b}^0 = {F}_{m,b}$.
    
    Our goal is to prove that \begin{equation}
        \sum_{\alpha = 0}^{2^b-1} {Q}_{m,b}^\alpha (\sigma_{1:m}, \tau_{1:m}) A^{\alpha\beta} (\sigma_{m+1}, \tau_{m+1}) = {Q}_{m+1,b}^\beta (\sigma_{1:m+1}, \tau_{1:m+1}), 
        \label{eq:AQFT-MPO}
    \end{equation}
    or ${Q}_{m,b} \, A  = {Q}_{m+1,b}$ in the shorthand of the preceding sections, in which the left-hand side indicates the result of attaching the core $A$ to the right of ${Q}_{m,b}$.
Direct computation reveals that ${Q}_{1,b}^\beta (\sigma, \tau) = e^{-\pi i (\sigma + \beta/2^b )\tau }$, which precisely matches the leftmost core from the statement of the theorem. Thus the claim~\eqref{eq:AQFT-MPO} will complete the proof of the theorem. We now turn to the proof of this claim.
    
    First, one can verify that for $x = \sum_{j=1}^n x_j 2^{-j}$ where $x_1,\ldots, x_n \in \{0,1\}$ and $n \geq b$, we have that 
    \begin{equation*}
        \chi^\alpha(x) = \prod_{k=1}^{b}\delta_{\alpha_{k}, x_{k}}.
    \end{equation*}
    Then since $(\sigma + u^\beta)/2 = \sigma/2 + \beta_{1}/4 + \cdots + \beta_{b}/2^{b+1}$, 
    \begin{equation*}
    \begin{split}
        &  \sum_{\alpha=0}^{2^b-1} B^\alpha_{m,b}(\tau_{1:m}) \ \chi^{\alpha}\left(\frac{\sigma_{m+1}+u^{\beta}}{2}\right) \\ 
        & \quad = \ \ \sum_{\alpha=0}^{2^b-1} e^{-2\pi i \sum_{j=1}^b \alpha_j 2^{-j} [0.\tau_{m}...\tau_{m-b+j}]} \ \chi^{\alpha}\left(\frac{\sigma_{m+1}+u^{\beta}}{2}\right)\\
        & \quad = \ \ \sum_{\alpha=0}^{2^b-1} e^{-2\pi i \sum_{j=1}^b \alpha_j 2^{-j} [0.\tau_{m}...\tau_{m-b+j}]} \ \delta_{\alpha_1, \sigma_{m+1}} \, \prod_{k=1}^{b-1}\delta_{\alpha_{k+1}, \beta_{k}}\\
        & \quad = \ \  e^{-2\pi i \sigma_{m+1} [0.0\tau_{m}...\tau_{m-b+1}]} \ e^{-2\pi i \sum_{j=2}^b \beta_{j-1} 2^{-j} [0.\tau_{m}...\tau_{m-b+j}]}
    \end{split}
    \end{equation*}
    The first factor in the last expression can be grouped with ${F}_{m,b} (\sigma_{1:m}, \tau_{1:m} ) $ and $e^{-\pi i \sigma_{m+1} \tau_{m+1}}$ to form ${F}_{m+1, b} (\sigma_{1:m+1},\tau_{1:m+1})$:
    \begin{equation*}
    \begin{split}
        &{F}_{m,b}(\sigma_{1:m},\tau_{1:m}) \, e^{-2\pi i \sigma_{m+1} [0.0\tau_{m}...\tau_{m-b+1}]} e^{-\pi i \sigma_{m+1} \tau_{m+1}} \\
        & \quad = \ \  {F}_{m,b} (\sigma_{1:m},\tau_{1:m}) \, e^{-2\pi i \sigma_{m+1} [0.\tau_{m+1}\tau_{m}...\tau_{m+1-b}]}\\
        & \quad = \ \  {F}_{m+1,b} (\sigma_{1:m+1},\tau_{1:m+1}),
    \end{split}   
    \end{equation*}
    and the second factor can be grouped with $e^{-\pi i u^\beta \tau_{m+1} }$ to form $B_{m+1, b}^\beta (\tau_{1:m+1})$:
    \begin{equation*}
    \begin{split}
        & e^{-2\pi i \sum_{j=2}^b \beta_{j-1} 2^{-j} [0.\tau_{m}...\tau_{m-b+j}]} \ e^{-\pi i u^\beta \tau_{m+1}} \\
        & \quad = \ \  e^{-2\pi i \sum_{j=1}^{b-1} \beta_{j} 2^{-j} [0.0\tau_{m}...\tau_{m+1-b+j}]} \ e^{-2\pi i \sum_{j=1}^{b} \beta_j 2^{-j} [0.\tau_{m+1}]} \\
        & \quad = \ \  e^{-2\pi i \sum_{j=1}^b \beta_{j} 2^{-j} [0.\tau_{m+1}\tau_{m}...\tau_{m+1-b+j}]} \\
        & \quad = \ \ B_{m+1, b}^\beta (\tau_{1:m+1}). 
    \end{split}
    \end{equation*}
Based on the definitions~\eqref{eq:A_AQFT} and~\eqref{eq:Qdef} of $A$ and $Q_{m,b}$, respectively, the claim~\eqref{eq:AQFT-MPO} follows.
\end{proof}

In fact we can derive the MPO core $A$ from the tensor diagram of the AQFT, by directly contracting local tensors in the AQFT circuit. This construction is shown in Fig.~\ref{fig:AQFT} in the case $b=2$.

\begin{figure}
    \centering   
    \includegraphics[width=0.7\textwidth]{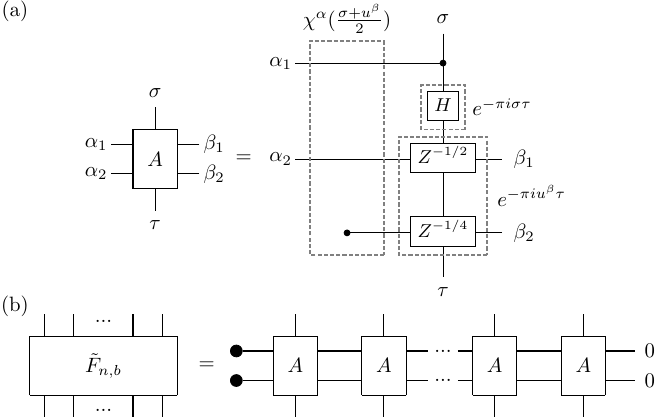}
    \caption{(a) The AQFT MPO core $A$ can be obtained by contracting local tensors in the AQFT circuit, as shown here for approximation level $b=2$. (b) Illustration of the AQFT MPO generated by $A$ for $b=2$.}
    \label{fig:AQFT}
\end{figure}

The entrywise error of this MPO relative to the true DFT/QFT can be bounded directly by the exact expression~\eqref{eq:AQFT} for entries, without relying on interpolation error bounds, as follows.

For any two phases $\theta_1, \theta_2 \in [0,2 \pi)$ separated by an angular distance $\Delta \theta \in [0, \pi]$, we always have 
\[
 |e^{-i\theta_1} - e^{-i\theta_2}| \leq \Delta \theta.
\]
Then for any fixed $\sigma_{1:n}, \tau_{1:n}$, the angular distance $\Delta \theta $ between the phases of the AQFT entry $F_{n,b} (\sigma_{1:n}, \tau_{1:n})$ and the corresponding QFT entry  $F (\sigma_{1:n}, \tau_{1:n} )$ is bounded as $\pi \sum_{k-l > b} 2^{-(k-l)}$. We can compute
\[
\sum_{k-l > b} 2^{-(k-l)} \leq \sum_{l=1}^n \sum_{k=l+b+1}^\infty 2^{-(k-l)} = n \sum_{j=b+1}^\infty 2^{-j} = n \, 2^{-b}.
\]
Therefore the entrywise error of the $n$-qubit AQFT with approximation level $b$ is bounded by $\pi n \,2^{-b}$.


We see that for a fixed number $n$ of qubits, the error of the AQFT decays as $O(1/K)$ where $K = 2^b$ is the bond dimension, and to maintain fixed error tolerance the bond dimension must scale linearly with $n$. On the other hand, in the Chebyshev-Lobatto construction of the QFT MPO,  the error decays super-exponentially with the bond dimension for a fixed number $n$ of qubits. Meanwhile, to maintain fixed error, the bond dimension need only scale sublinearly with $n$. Thus the piecewise constant interpolation used implicitly by the AQFT yields a less efficient MPO approximation for the QFT. On the other hand, the AQFT is presented quite naturally as a quantum circuit. Whether other interpolative constructions can yield more efficient quantum circuits that approximate the QFT is an interesting open question.

\printbibliography

@book{TrefethenBook2019,
	address = {Philadelphia, PA, USA},
	author = {Trefethen, Lloyd N.},
	publisher = {SIAM-Society for Industrial and Applied Mathematics},
	title = {Approximation Theory and Approximation Practice, Extended Edition},
	year = {2019}}

@article{Chen_2023_QFT,
  title = {Quantum Fourier Transform Has Small Entanglement},
  author = {Chen, Jielun and Stoudenmire, E.M. and White, Steven R.},
  journal = {PRX Quantum},
  volume = {4},
  issue = {4},
  pages = {040318},
  numpages = {31},
  year = {2023},
  month = {Oct},
  publisher = {American Physical Society},
  doi = {10.1103/PRXQuantum.4.040318},
  url = {https://link.aps.org/doi/10.1103/PRXQuantum.4.040318}
}

@article{Dolgov_2012,
	author = {Dolgov, Sergey and Khoromskij, Boris and Savostyanov, Dmitry},
	journal = {Journal of Fourier Analysis and Applications},
	number = {5},
	pages = {915--953},
	title = {Superfast Fourier Transform Using QTT Approximation},
	volume = {18},
	year = {2012}}

@misc{Coppersmith_2002,
      title={An approximate Fourier transform useful in quantum factoring}, 
      author={D. Coppersmith},
      year={2002},
      eprint={quant-ph/0201067},
      archivePrefix={arXiv},
      primaryClass={quant-ph}
}

@misc{Aharonov_2007,
      title={The quantum FFT can be classically simulated}, 
      author={Dorit Aharonov and Zeph Landau and Johann Makowsky},
      year={2007},
      eprint={quant-ph/0611156},
      archivePrefix={arXiv},
      primaryClass={quant-ph}
}

@article{Yoran_2007,
   title={Efficient classical simulation of the approximate quantum Fourier transform},
   volume={76},
   ISSN={1094-1622},
   url={http://dx.doi.org/10.1103/PhysRevA.76.042321},
   DOI={10.1103/physreva.76.042321},
   number={4},
   journal={Physical Review A},
   publisher={American Physical Society (APS)},
   author={Yoran, Nadav and Short, Anthony J.},
   year={2007},
   month=oct }

@article{Browne_2007,
doi = {10.1088/1367-2630/9/5/146},
url = {https://dx.doi.org/10.1088/1367-2630/9/5/146},
year = {2007},
month = {may},
publisher = {},
volume = {9},
number = {5},
pages = {146},
author = {Daniel E Browne},
title = {Efficient classical simulation of the quantum Fourier transform},
journal = {New Journal of Physics}
}

@article{Shinaoka_2023,
  title = {Multiscale Space-Time Ansatz for Correlation Functions of Quantum Systems Based on Quantics Tensor Trains},
  author = {Shinaoka, Hiroshi and Wallerberger, Markus and Murakami, Yuta and Nogaki, Kosuke and Sakurai, Rihito and Werner, Philipp and Kauch, Anna},
  journal = {Phys. Rev. X},
  volume = {13},
  issue = {2},
  pages = {021015},
  numpages = {27},
  year = {2023},
  month = {Apr},
  publisher = {American Physical Society},
  doi = {10.1103/PhysRevX.13.021015},
  url = {https://link.aps.org/doi/10.1103/PhysRevX.13.021015}
}

@misc{Woolfe_2014,
      title={Scale invariance and efficient classical simulation of the quantum Fourier transform}, 
      author={Kieran J. Woolfe and Charles D. Hill and Lloyd C. L. Hollenberg},
      year={2014},
      eprint={1406.0931},
      archivePrefix={arXiv},
      primaryClass={quant-ph}
}

@article{oseledets2011tensortrain,
  title={Tensor-train decomposition},
  author={Oseledets, Ivan V},
  journal={SIAM Journal on Scientific Computing},
  volume={33},
  number={5},
  pages={2295--2317},
  year={2011},
  publisher={SIAM}
}

@article{Boyd_1992,
    title = {A fast algorithm for Chebyshev, Fourier, and sinc interpolation onto an irregular grid},
    journal = {Journal of Computational Physics},
    volume = {103},
    number = {2},
    pages = {243-257},
    year = {1992},
    issn = {0021-9991},
    doi = {https://doi.org/10.1016/0021-9991(92)90399-J},
    url = {https://www.sciencedirect.com/science/article/pii/002199919290399J},
    author = {John P Boyd}
}

@article{Li_2015_butterfly,
author = {Li, Yingzhou and Yang, Haizhao and Martin, Eileen R. and Ho, Kenneth L. and Ying, Lexing},
title = {Butterfly Factorization},
journal = {Multiscale Modeling \& Simulation},
volume = {13},
number = {2},
pages = {714-732},
year = {2015},
doi = {10.1137/15M1007173},
URL = {https://doi.org/10.1137/15M1007173},
eprint = { https://doi.org/10.1137/15M1007173}
}

@misc{lindsey2023multiscale,
      title={Multiscale interpolative construction of quantized tensor trains}, 
      author={Michael Lindsey},
      year={2023},
      eprint={2311.12554},
      archivePrefix={arXiv},
      primaryClass={math.NA}
}

@article{Dolgov_2020_cross,
title = {Parallel cross interpolation for high-precision calculation of high-dimensional integrals},
journal = {Computer Physics Communications},
volume = {246},
pages = {106869},
year = {2020},
issn = {0010-4655},
doi = {https://doi.org/10.1016/j.cpc.2019.106869},
url = {https://www.sciencedirect.com/science/article/pii/S0010465519302565},
author = {Sergey Dolgov and Dmitry Savostyanov}
}

@article{Savostyanov_2014_cross,
title = {Quasioptimality of maximum-volume cross interpolation of tensors},
journal = {Linear Algebra and its Applications},
volume = {458},
pages = {217-244},
year = {2014},
issn = {0024-3795},
doi = {https://doi.org/10.1016/j.laa.2014.06.006},
url = {https://www.sciencedirect.com/science/article/pii/S0024379514003711},
author = {Dmitry V. Savostyanov}
}

@misc{khoromskij2014tensor,
      title={Tensor Numerical Methods for High-dimensional PDEs: Basic Theory and Initial Applications}, 
      author={Boris N. Khoromskij},
      year={2014},
      eprint={1408.4053},
      archivePrefix={arXiv},
      primaryClass={math.NA}
}

@article{Gourianov_2022,
   title={A quantum-inspired approach to exploit turbulence structures},
   volume={2},
   ISSN={2662-8457},
   url={http://dx.doi.org/10.1038/s43588-021-00181-1},
   DOI={10.1038/s43588-021-00181-1},
   number={1},
   journal={Nature Computational Science},

   publisher={Springer Science and Business Media LLC},
   author={Gourianov, Nikita and Lubasch, Michael and Dolgov, Sergey and van den Berg, Quincy Y. and Babaee, Hessam and Givi, Peyman and Kiffner, Martin and Jaksch, Dieter},
   year={2022},
   month=jan, 
    pages={30-37} }

@misc{kiffner2023tensor,
      title={Tensor network reduced order models for wall-bounded flows}, 
      author={Martin Kiffner and Dieter Jaksch},
      year={2023},
      eprint={2303.03010},
      archivePrefix={arXiv},
      primaryClass={physics.flu-dyn}
}

@article{Ye_2022_vlasov,
  title = {Quantum-inspired method for solving the Vlasov-Poisson equations},
  author = {Ye, Erika and Loureiro, Nuno F. G.},
  journal = {Phys. Rev. E},
  volume = {106},
  issue = {3},
  pages = {035208},
  numpages = {20},
  year = {2022},
  month = {Sep},
  publisher = {American Physical Society},
  doi = {10.1103/PhysRevE.106.035208},
  url = {https://link.aps.org/doi/10.1103/PhysRevE.106.035208}
}

@article{Daley_2004,
    doi = {10.1088/1742-5468/2004/04/P04005},
    url = {https://dx.doi.org/10.1088/1742-5468/2004/04/P04005},
    year = {2004},
    month = {apr},
    publisher = {},
    volume = {2004},
    number = {04},
    pages = {P04005},
    author = {A J Daley and  C Kollath and  U Schollwöck and  G Vidal},
    title = {Time-dependent density-matrix renormalization-group using adaptive effective Hilbert
    spaces},
    journal = {Journal of Statistical Mechanics: Theory and Experiment},
}

@article{White_2004,
  title = {Real-Time Evolution Using the Density Matrix Renormalization Group},
  author = {White, Steven R. and Feiguin, Adrian E.},
  journal = {Phys. Rev. Lett.},
  volume = {93},
  issue = {7},
  pages = {076401},
  numpages = {4},
  year = {2004},
  month = {Aug},
  publisher = {American Physical Society},
  doi = {10.1103/PhysRevLett.93.076401},
  url = {https://link.aps.org/doi/10.1103/PhysRevLett.93.076401}
}

@article{Vidal_2007,
  title = {Classical Simulation of Infinite-Size Quantum Lattice Systems in One Spatial Dimension},
  author = {Vidal, G.},
  journal = {Phys. Rev. Lett.},
  volume = {98},
  issue = {7},
  pages = {070201},
  numpages = {4},
  year = {2007},
  month = {Feb},
  publisher = {American Physical Society},
  doi = {10.1103/PhysRevLett.98.070201},
  url = {https://link.aps.org/doi/10.1103/PhysRevLett.98.070201}
}

@article{Edelman1997TheFF,
  title={The Future Fast Fourier Transform?},
  author={Alan Edelman and Peter McCorquodale and Sivan Toledo},
  journal={SIAM J. Sci. Comput.},
  year={1997},
  volume={20},
  pages={1094-1114},
  url={https://api.semanticscholar.org/CorpusID:1837696}
}

@article{Fannes_1992,
    author={Fannes, M.
    and Nachtergaele, B.
    and Werner, R. F.},
    title={Finitely correlated states on quantum spin chains},
    journal={Communications in Mathematical Physics},
    year={1992},
    month={Mar},
    day={01},
    volume={144},
    number={3},
    pages={443-490},
    issn={1432-0916},
    doi={10.1007/BF02099178},
    url={https://doi.org/10.1007/BF02099178}
}

@article{Klumper_1992,
    author={Kl{\"u}mper, A.
    and Schadschneider, A.
    and Zittartz, J.},
    title={Groundstate properties of a generalized VBS-model},
    journal={Zeitschrift f{\"u}r Physik B Condensed Matter},
    year={1992},
    month={Oct},
    day={01},
    volume={87},
    number={3},
    pages={281-287},
    issn={1431-584X},
    doi={10.1007/BF01309281},
    url={https://doi.org/10.1007/BF01309281}
}

@article{Ostlund_1995,
  title = {Thermodynamic Limit of Density Matrix Renormalization},
  author = {\"Ostlund, Stellan and Rommer, Stefan},
  journal = {Phys. Rev. Lett.},
  volume = {75},
  issue = {19},
  pages = {3537--3540},
  numpages = {0},
  year = {1995},
  month = {Nov},
  publisher = {American Physical Society},
  doi = {10.1103/PhysRevLett.75.3537},
  url = {https://link.aps.org/doi/10.1103/PhysRevLett.75.3537}
}

@article{Vidal_2003,
  title = {Efficient Classical Simulation of Slightly Entangled Quantum Computations},
  author = {Vidal, Guifr\'e},
  journal = {Phys. Rev. Lett.},
  volume = {91},
  issue = {14},
  pages = {147902},
  numpages = {4},
  year = {2003},
  month = {Oct},
  publisher = {American Physical Society},
  doi = {10.1103/PhysRevLett.91.147902},
  url = {https://link.aps.org/doi/10.1103/PhysRevLett.91.147902}
}

@article{Pirvu_2010,
   title={Matrix product operator representations},
   volume={12},
   ISSN={1367-2630},
   url={http://dx.doi.org/10.1088/1367-2630/12/2/025012},
   DOI={10.1088/1367-2630/12/2/025012},
   number={2},
   journal={New Journal of Physics},
   publisher={IOP Publishing},
   author={Pirvu, B and Murg, V and Cirac, J I and Verstraete, F},
   year={2010},
   month=feb, pages={025012}
}

@article{Khoromskij_2009_QTT,
    author = {Khoromskij, Boris},
    year = {2009},
    month = {01},
    pages = {},
    title = {O ( d log N )-Quantics Approximation of N - d Tensors in High-Dimensional Numerical Modeling},
    volume = {34},
    journal = {Constructive Approximation - CONSTR APPROX},
    doi = {10.1007/s00365-011-9131-1}
}

@INPROCEEDINGS{Shor_1994,
  author={Shor, P.W.},
  booktitle={Proceedings 35th Annual Symposium on Foundations of Computer Science}, 
  title={Algorithms for quantum computation: discrete logarithms and factoring}, 
  year={1994},
  volume={},
  number={},
  pages={124-134},
  doi={10.1109/SFCS.1994.365700}
}
\end{document}